\documentclass[conference]{IEEEtran}

\usepackage[utf8]{inputenc}
\usepackage[margin=0.545in]{geometry}
\usepackage[english]{babel}  % better hyphenation
\usepackage[T1]{fontenc} %T2A for Cyrillic symbols
\usepackage{epsfig}
\usepackage[cmex10]{amsmath}
\usepackage{amssymb, amsbsy, amsthm}
\usepackage{mathrsfs}
\usepackage{xspace}          % for \xspace = possible space after commands
\usepackage{color}
\usepackage{url}             % hyperlinks
\usepackage{booktabs} %beautiful tables: toprule, midrule, bottomrule
\usepackage[colorlinks=true,citecolor=blue,linkcolor=blue]{hyperref}
\usepackage[linesnumbered,ruled,vlined,titlenumbered]{algorithm2e}
\usepackage{dblfloatfix} % Fixes a bug with double column floats (now they can be placed at the bottom of a page)

\IEEEoverridecommandlockouts

\newtheorem{theorem}{Theorem}
\newtheorem{lemma}[theorem]{Lemma}

\newtheorem{definition}[theorem]{Definition}

\usepackage{varioref}        % for \vref{:...} gives "As 1.2 on page 4"
\usepackage{fancyref}

\makeatletter
\def\mkfancyprefix#1#2{%
\expandafter\def\csname fancyref#1labelprefix\endcsname{#1}%
\begingroup\def\x{\endgroup\frefformat{plain}}%
    \expandafter\x\csname fancyref#1labelprefix\endcsname
    {\MakeLowercase{#2}\fancyrefdefaultspacing##1}%
\begingroup\def\x{\endgroup\Frefformat{plain}}%
    \expandafter\x\csname fancyref#1labelprefix\endcsname
    {#2\fancyrefdefaultspacing##1}%
\begingroup\def\x{\endgroup\frefformat{vario}}%
    \expandafter\x\csname fancyref#1labelprefix\endcsname
    {\MakeLowercase{#2}\fancyrefdefaultspacing##1##3}%
\begingroup\def\x{\endgroup\Frefformat{vario}}%
    \expandafter\x\csname fancyref#1labelprefix\endcsname
    {#2\fancyrefdefaultspacing##1##3}%
}
\makeatother
\fancyrefchangeprefix{\fancyrefeqlabelprefix}{eqn}
\mkfancyprefix{sec}{Section}
\mkfancyprefix{subsec}{Section}
\mkfancyprefix{tbl}{Table}
\mkfancyprefix{thm}{Theorem}
\mkfancyprefix{lem}{Lemma}
\mkfancyprefix{cor}{Corollary}
\mkfancyprefix{prop}{Proposition}
\mkfancyprefix{prob}{Problem}
\mkfancyprefix{alg}{Algorithm}
\mkfancyprefix{inv}{Invariant}
\mkfancyprefix{ex}{Example}
\mkfancyprefix{line}{Line}
\mkfancyprefix{def}{Definition}
\mkfancyprefix{itm}{Item}
\newcommand{\cref}[1]{\Fref{#1}}

\makeatletter
\newcommand{\removelatexerror}{\let\@latex@error\@gobble}
\makeatother

\newcommand{\printalgoIEEE}[1]
{{\centering
\scalebox{0.97}{
\removelatexerror
\begin{tabular}{p{\columnwidth}}
\begin{algorithm}[H]
 \begin{small}
 #1
 \end{small}
\end{algorithm}
\end{tabular}
}
}
}

\newcommand{\printalgoIEEEdoublecolumn}[1]{
{\centering
\begin{table*}[!t]
\scalebox{0.97}{
\removelatexerror
\begin{tabular}{p{\textwidth}}
\begin{algorithm}[H]
 #1
\end{algorithm}
\end{tabular}
}
\end{table*}
}
}

\DeclareMathOperator{\rank}{rank}

\DeclareUnicodeCharacter{0109}{\koppa}

\renewcommand{\vec}[1]{\ensuremath{\mathbf{#1}}}
\newcommand{\Mat}[1]{\ensuremath{{#1}}}

\newcommand{\LEEAOutputRx}{\ensuremath{r_\mathrm{out}}}
\newcommand{\LEEAOutputUx}{\ensuremath{u_\mathrm{out}}}

\newcommand{\LEEAOutputVx}{\ensuremath{v_\mathrm{out}}}

\newcommand{\Normelement}{\beta}

\newcommand{\mycode}[1]{\ensuremath{\mathcal{#1}}}
\newcommand{\Gab}[1]{\ensuremath{\mycode{G}[#1]}}

\newcommand{\N}{\mathbb{N}}
\newcommand{\Fq}{\mathbb{F}_q}
\newcommand{\Fqm}{\mathbb{F}_{q^m}}

\newcommand{\qdeg}{\mathrm{deg}_q}
\newcommand{\ext}{\mathrm{ext}}
\newcommand{\extinv}{\mathrm{ext}^{-1}}

\newcommand{\LH}[1]{\langle #1 \rangle}
\newcommand{\Lset}{\mathcal{L}_{q^m}}
\newcommand{\BigO}[1]{\mathcal{O}\left(#1\right)}
\newcommand{\BigOtext}[1]{\mathcal{O}(#1)}

\newcommand{\OMul}[1]{\mathcal{M}_{q^m}\left(#1\right)}

\newcommand{\OMulSkew}[1]{\mathcal{M}_{q^m}\left(#1\right)}
\newcommand{\ODiv}[1]{\mathcal{D}_{q^m}\left(#1\right)}
\newcommand{\ODivSkew}[1]{\mathcal{D}_{q^m}\left(#1\right)}

\newcommand{\OMSP}[1]{\mathcal{MSP}_{q^m}\left(#1\right)}
\newcommand{\OMPE}[1]{\mathcal{MPE}_{q^m}\left(#1\right)}
\newcommand{\OIP}[1]{\mathcal{I}_{q^m}\left(#1\right)}

\newcommand{\Lsetmaxs}{\Lset^{\leq s}}

\newcommand{\qtr}[1]{\hat{#1}}
\newcommand{\mul}{\cdot}

\newcommand{\sstar}{s^\ast}

\newcommand{\Lsmallers}{\Lset^{< s}}
\newcommand{\Lsmallerk}{\Lset^{< k}}

\newcommand{\rk}{\mathrm{rk}}

\newcommand{\MSP}[1]{\mathcal{M}_{#1}}
\newcommand{\U}{\mathcal{U}}

\newcommand{\AMSP}[1]{\mathrm{MSP}\left( #1 \right)}
\newcommand{\AMPE}[2]{\mathrm{MPE}\left( #1,#2 \right)}
\newcommand{\AIP}[1]{\mathrm{IP}\left( #1 \right)}
\newcommand{\ARI}[1]{\mathrm{RightInv}\left( #1 \right)}
\newcommand{\ADIV}[1]{\mathrm{RightDiv}\left( #1 \right)}
\newcommand{\ARDIV}[1]{\ADIV{#1}}
\newcommand{\ALDIV}[1]{\mathrm{LeftDiv}\left( #1 \right)}
\newcommand{\ALEEA}[1]{\mathrm{HalfLEEA}\left( #1 \right)}

\newcommand{\RSpace}[1]{\ker(#1)}
\newcommand{\quo}{\chi}
\newcommand{\remainder}{\varrho}
\newcommand{\qL}{\quo_\mathrm{L}}
\newcommand{\qR}{\quo_\mathrm{R}}
\newcommand{\rL}{\remainder_\mathrm{L}}
\newcommand{\rR}{\remainder_\mathrm{R}}
\newcommand{\gauss}[1]{\lfloor #1 \rfloor}

\newcommand{\IP}[1]{\mathcal{I}_{#1}}
\newcommand{\autom}{\sigma}
\newcommand{\automi}[1]{\sigma^{#1}}

\newcommand{\Aset}{\Fqm[x;\autom]} %\mathcal{L}_\autom}
\newcommand{\Asetmaxs}{\Aset_{\leq s}}
\newcommand{\Asetinv}{\Fqm[x;\autom^{-1}]} %\mathcal{L}_\autom}
\newcommand{\Asetinvmaxs}{\Asetinv_{\leq s}}
\newcommand{\amul}{\cdot}

\definecolor{darkred}{rgb}{0.7,0,0}

\renewcommand{\tilde}{\widetilde}

\newcommand{\smallsum}{{\textstyle\sum\nolimits}}

\renewcommand{\c}{\vec c}
\renewcommand{\r}{\vec r}
\renewcommand{\a}{\vec a}
\newcommand{\e}{\vec e}
\newcommand{\BM}{\vec B}

\newcommand{\Eexp}{^\mathrm{E}}
\newcommand{\Rexp}{^\mathrm{R}}
\newcommand{\Cexp}{^\mathrm{C}}

\newcommand{\gC}{\Gamma\Cexp}
\newcommand{\gCbar}{\overline{\gC}}
\newcommand{\gCtilde}{\tilde{\gC}}
\newcommand{\lR}{\Lambda\Rexp}
\newcommand{\lE}{\Lambda\Eexp}

\newcommand{\BE}{\BM\Eexp}
\newcommand{\BR}{\BM\Rexp}
\newcommand{\BC}{\BM\Cexp}
\newcommand{\aE}{\a\Eexp}
\newcommand{\aR}{\a\Rexp}
\newcommand{\aC}{\a\Cexp}

\newcommand{\tE}{\tau}
\newcommand{\tR}{\varrho}
\newcommand{\tC}{\gamma}

\newcommand{\mymod}{\; \mathrm{mod} \,}

\newcommand{\rout}{\LEEAOutputRx}
\newcommand{\uout}{\LEEAOutputUx}
\newcommand{\vout}{\LEEAOutputVx}

\begin{document}

\title{Sub-Quadratic Decoding of Gabidulin Codes}
\author{\IEEEauthorblockN{Sven Puchinger$^1$ and Antonia Wachter-Zeh$^2$}\\
\IEEEauthorblockA{
$^1$Institute of Communications Engineering, Ulm University, Ulm, Germany\\
$^2$Department of Computer Science,
Technion---Israel Institute of Technology, Haifa, Israel\\
Email: \emph{sven.puchinger@uni-ulm.de, antonia@cs.technion.ac.il}
}
\thanks{Sven Puchinger was supported by the German Research Foundation (DFG), grant BO~867/29-3.
Antonia Wachter-Zeh's work was supported by the European Union’s Horizon 2020 research and innovation programme under the Marie Sklodowska-Curie grant agreement No. 655109.}
}

\maketitle

\begin{abstract}
This paper shows how to decode errors and erasures with Gabidulin codes in sub-quadratic time in the code length, improving previous algorithms which had at least quadratic complexity.
The complexity reduction is achieved by accelerating operations on linearized polynomials.
In particular, we present fast algorithms for division, multi-point evaluation and interpolation of linearized polynomials and show how to efficiently compute minimal subspace polynomials.
\end{abstract}


%%%%%%%%%%%%%%%%%%%%%%%%%%%%%%%%%%%%%%%%%%%%%%%%%%%%%%%%%%%%%%%%%%%%%%%%%%%%%%
% Intro
%%%%%%%%%%%%%%%%%%%%%%%%%%%%%%%%%%%%%%%%%%%%%%%%%%%%%%%%%%%%%%%%%%%%%%%%%%%%%%

\section{Introduction}

Rank-metric codes can be found in a wide range of applications, including network coding~\cite{silva_rank_metric_approach}, code-based cryptosystems~\cite{Loidreau2010Designing}, and distributed storage systems~\cite{SilbersteinRawatVish-ErrorResilDistributedStorage_2012}.
A rank-metric code is a set of matrices and the distance between any two codewords (i.e., matrices) is the rank of the difference of the two matrices. 
Gabidulin codes are the analog of Reed--Solomon codes in the rank metric. They are defined by evaluating linearized polynomials at linearly independent points of an extension field~$\Fqm$.

In this paper, we recall that the complexity of error and erasure decoding of Gabidulin codes is determined by the complexity of the operations multiplication, division, multi-point evaluation with linearized polynomials and the calculation of minimal subspace polynomials.
The multiplication of two linearized polynomials of degree at most $s$ is known to be in $\BigOtext{s^{1.69}}$ over~$\Fqm$~\cite{wachter2013decoding}.
However, the division of two linearized polynomials was so far believed to be in~$\BigOtext{s^2}$, compare~\cite{Gadouleau_complexity}.
We show that the reduction of linearized polynomial division to \emph{skew polynomial} multiplication in \cite{caruso2012some} implies a sub-quadratic division algorithm by generalizing the above mentioned multiplication algorithm to skew polynomials.
Finding a minimal subspace polynomial and performing a multi-point evaluation were both known to have complexity $\BigOtext{s^2}$, see~\cite{li2014transform}, and the interpolation $\BigOtext{s^3}$.
We also present fast methods for these operations.

The papers \cite{SilvaKschischang-FastEncodingDecodingGabidulin-2009} and \cite{WachterAfanSido-FastDecGabidulin_DCC} consider fast decoding strategies of Gabidulin codes over $\Fq$ and the complexity of several steps of decoding Gabidulin codes is reduced to $\BigOtext{n^3}$ operations over~$\Fq$.
We show that our algorithms improve these results when considered over $\Fq$.
Hence, to our knowledge, this paper is the first work which achieves sub-quadratic decoding complexity over~$\Fqm$.

An extended version of this paper was submitted to the Journal of Symbolic Computation \cite{puchinger2015fast}, concentrating on the fast operations and their optimality.
Here, we summarize the results of \cite{puchinger2015fast}, skipping several technical proofs, and describe the connection to the decoding problem more comprehensively.

%%%%%%%%%%%%%%%%%%%%%%%%%%%%%%%%%%%%%%%%%%%%%%%%%%%%%%%%%%%%%%%%%%%%%%%%%%%%%%
% Fundamentals
%%%%%%%%%%%%%%%%%%%%%%%%%%%%%%%%%%%%%%%%%%%%%%%%%%%%%%%%%%%%%%%%%%%%%%%%%%%%%%
\section{Preliminaries}\label{sec:preliminaries}

%\subsection{Notations}
Let $q$ be a prime power, $\Fq$ be a finite field with $q$ elements and $\Fqm$ an extension extension field of $\Fq$. Since $\Fqm$ can be seen as an $m$-dimensional vector space over $\Fq$, there is a vector space isomorphism $\ext : \Fqm^n \mapsto \Fq^{m \times n}$ with inverse $\extinv$.
%We denote the $q$-power of a monomial by $x^{[i]}= x^{q^i}$ for any integer $i$.
A subspace of $\Fqm$ is always meant with respect to $\Fq$ as the scalar field.
For $A \subseteq \Fqm$, $\LH{A}$ is the $\Fq$-span of $A$.
By $\omega$ we denote the matrix multiplication complexity exponent, e.g., $\omega \approx 2.376$ in the Coppersmith--Winograd algorithm.

\subsection{Linearized Polynomials}

A \emph{linearized polynomial} \cite{ore1933special} is a polynomial of the form
\begin{align*}
a = \smallsum_{k=0}^{d_a} a_k x^{q^k} = \smallsum_{k=0}^{d_a} a_k x^{[k]}, \quad a_k \in \Fqm, \quad a_{d_a} \neq 0
\end{align*}
with $[i] := q^i$, where $d_a \in \N_0 \cup \{-\infty\}$ is the $q$-degree $\qdeg a$.
The set of all linearized polynomials for given $q$ and $m$ is denoted by $\Lset$.
The addition $+$ in $\Lset$ is defined as for ordinary polynomials and the multiplication $\mul$ as
\begin{align*}
a \mul b = \smallsum_{i} \big( \smallsum_{j=0}^{i} a_j b_{i-j}^{[j]} \big) x^{[i]}.
\end{align*}
Note that if $\Lset$ is seen as a subset of $\Fqm[x]$, the multiplication~$\mul$ equals the composition of two polynomials. 
It is shown in \cite{ore1933special} that $(\Lset,+,\mul)$ is a (non-commutative) ring with multiplicative identity $x^{[0]}=x$.
For $s \in \N$, we define $\Lsetmaxs := \{a \in \Lset \, : \, \qdeg a \leq s\}$ and $\Lsmallers$ analogously.
A polynomial $a$ is called monic if $a_{\qdeg a} = 1$.
It is easy to see that $\qdeg(a \mul b) = \qdeg a + \qdeg b$ and $\qdeg(a+b) \leq \max\{\qdeg a, \qdeg b\}$.

For $a \in \Lset$, we define an evaluation map
\begin{align*}
a(\cdot) \, : \, \Fqm \to \Fqm \; , \; \alpha \mapsto a(\alpha) = \smallsum_{i} a_i \alpha^{[i]},
\end{align*}
which is an $\Fq$-linear for any $a \in \Lset$.
Thus, the root space $\RSpace{a} = \{\alpha \in \Fqm \, : \, a(\alpha) = 0\}$ is a subspace of $\Fqm$.
It is also clear that $(a \mul b)(\alpha) = a(b(\alpha))$.
$\Lset$ is a left and right Euclidean domain as shown by the following lemma.
\begin{lemma}[\cite{ore1933special}]\label{lem:division}
For $a,b \in \Lset$, there exist unique polynomials $\qR,\qL \in \Lset$ (quotients) and $\rR,\rL \in \Lset$ (remainders) such that
$a = \qR \mul b + \rR$ (right division) and $a = b \mul \qL + \rL$ (left division),
where $\qdeg \rR  < \qdeg b$ and $\qdeg \rL < \qdeg b$.
\end{lemma}

Lemma \ref{lem:division} allows us to define a (right) modulo operation on $\Lset$ such that $a \equiv b \mod c$ if $\exists$ $d \in \Lset$ such that $a = b + d \mul c$.
In the following, we use this definition of "mod".

Division also immediately gives us a linearized equivalent to the \emph{Extended Euclidean algorithm} (LEEA).
In \cite{wachter2013decoding}, a LEEA with stopping condition is presented such that for $a,b \in \Lset$, $d_\mathrm{stop} \in \mathbb{N}$, $[\rout,\uout,\vout] \gets \ALEEA{a,b,d_\mathrm{stop}}$ outputs polynomials with $\rout = \uout \mul a + \vout \mul b$ is $\rout$ is the first remainder appearing in the LEEA with $\deg \rout < d_\mathrm{stop}$.

\emph{Minimal subspace polynomials} are special linearized polynomials, with the property that their $q$-degree is equal to their number of linearly independent roots.
\begin{lemma}[\cite{wachter2013decoding}]\label{lem:MSP}
Let $\U$ be a subspace of $\Fqm$.
Then there is a unique nonzero monic polynomial $\MSP{\U} \in \Lset$ of minimal degree $\qdeg \MSP{\U} =\dim \U$ such that $\ker(\MSP{\U}) = \U$.
$\MSP{\U}$ is called minimal subspace polynomial (MSP) of $\U$.
\end{lemma}

\emph{Multi-point evaluation} (MPE) is the process of evaluating a polynomial $a \in \Lset$ at multiple points.
The dual problem is called \emph{interpolation} and based on the following lemma.

\begin{lemma}[\cite{silva2007rank}]\label{lem:interpolation_existence}
Let $(x_1,y_1),\dots,(x_s,y_s) \in \Fqm^2$, linearly independent $x_i$'s. Then there exists a unique interpolation polynomial $\IP{\{(x_i,y_i)\}_{i=1}^{s}} \in \Lsmallers$ such that $\IP{\{(x_i,y_i)\}_{i=1}^{s}}(x_i) = y_i$ for all $i$.
\end{lemma}

\subsection{Skew Polynomials}

The ring of \emph{skew polynomials} $\Fqm[x;\autom]$ \cite{ore1933theory} with automorphism $\autom$, is defined as the set of polynomials $\sum_i a_i x^i$, $a_i \in \Fqm$, with multiplication rule $xa = \autom(a) x$ $\forall a \in \Fqm$ and ordinary component-wise addition.
The degree is defined as usual.
$\Fqm[x;\autom]$ is left and right Euclidean, i.e., Lemma \ref{lem:division} also holds for skew polynomials.
There is a ring isomorphism $\varphi : \Lset \to \Fqm[x;\cdot^q], \, \sum_{i} a_i x^{[i]} \mapsto \sum_{i} a_i x^i$, where $\cdot^q$ is the \emph{Frobenius automorphism}.
We utilize this fact to obtaining fast algorithms for linearized polynomials.

\printalgoIEEEdoublecolumn{
\DontPrintSemicolon
\KwIn{Received word $\r \in \Fqm^n$; 
$(g_1,\dots,g_n) = (\beta^{[0]},\dots,\beta^{[n-1]})$ normal basis of $\Fqm$ over $\Fq$;
$\aR \in \Fqm^\tR$;
$\BC \in \Fq^{\tC \times n}$
}
\KwOut{Estimated evaluation polynomial $f$ with $\qdeg f < k$ or ``decoding failure''.}
$d_i\Cexp \leftarrow \smallsum_{j=1}^{n}\BC_{i,j}\Normelement^{[j-1]}$ for all $i=1,\dots,\tC$ \hfill \tcp{negligible} \label{line:gao_ee_a}
$\gC \leftarrow \MSP{\LH{d_1\Cexp, \dots, d_\tC\Cexp}}$; $\,$
$\lR \leftarrow \MSP{\LH{a_1\Rexp, \dots, a_\tR\Rexp}}$; $\,$
Calculate $\gCtilde$ as in \eqref{eq:gCtilde} using \eqref{eq:gCbar}. \hfill \tcp{$2 \cdot \OMSP{n} + \BigO{n}$}  \label{line:gao_ee_d}
$\qtr{r} \leftarrow \IP{\{(g_i,r_i)\}_{i=1}^{n}}$; $\,$
$\qtr{y} \gets \lR \cdot \qtr{r} \cdot \gCtilde \mod (x^{[m]}-x^{[0]})$ \hfill \tcp{$\OIP{n} + \ODiv{n}$}  \label{line:gao_ee_f}
$[\rout, \uout, \vout] \gets \ALEEA{\qtr{y}, x^{[m]}-x^{[0]}, \lfloor\tfrac{n+k+\tR+\tC}{2}\rfloor }$ \hfill \tcp{$\ODiv{n} \log(n)$}  \label{line:gao_ee_g}
$[\qL, \rL] \gets \ALDIV{\rout, \uout \mul \lR}$; $\,$
$[\qR, \rR] \gets \ARDIV{\qL, \gCtilde \mod (x^{[m]}-x^{[0]})}$ \hfill \tcp{$2 \mul \ODiv{n}$}  \label{line:gao_ee_i}
\lIf{$\rL = 0$ \textup{\textbf{and}} $\rR = 0$}{\Return{$f \gets \qL$}}
\lElse{\Return{``decoding failure''}}
\caption{$\mathrm{DecodeGaoGabidulinErasures}\big(\r, (g_1,g_2,\dots,g_n), \aR, \BC \big)$ %\cite{wachter2013decoding}}
{\cite[Algorithm~3.7]{wachter2013decoding}}}
\label{alg:gaoalgo_ee}
}

\subsection{Rank-Metric and Gabidulin Codes}
Codes in the rank-metric are a set of matrices over some finite field $\Fq$ and the \emph{rank distance} between two matrices is defined to be the rank of their difference. 
Using the mapping $\ext$, there is a bijection between any matrix in $\Fq^{m \times n }$ and a vector in $\Fqm^n$. 
By slight abuse of notation, we use $\rk(\vec{a}) := \rank(\ext^{-1}(\Mat{A}))$, where $\vec{a} \in \Fqm^n$ and $\Mat{A}\in \Fq^{m \times n}$.
The \emph{minimum rank distance} $d_R$ of a block code $\mycode{C} \subseteq \Fqm^n$ is
\begin{equation*}
\mathrm{d}_\mathrm{R} = \min \big\lbrace \rk(\vec{c}_1-\vec{c}_2) \; : \; \vec{c}_1,\vec{c}_2 \in \mycode{C}, \vec{c}_1 \neq \vec{c}_2\big\rbrace. 
\end{equation*}

\emph{Gabidulin codes} \cite{Delsarte_1978,Gabidulin_TheoryOfCodes_1985,Roth_RankCodes_1991} are a special class of MRD codes, i.e. $\mathrm{d}_\mathrm{R} = n-k+1$, and 
are considered as the analogs of Reed--Solomon codes in rank metric.
They can be defined by the evaluation of degree-restricted linearized polynomials.

\begin{definition}[Gabidulin Code, \cite{Gabidulin_TheoryOfCodes_1985}]\label{def:gabidulin_code}
Fix $g_1, \dots, g_n \in \Fqm$, linearly independent over $\Fq$.
A linear Gabidulin code $\Gab{n,k}$ over $\Fqm$ of length $n\leq m$ and dimension $k \leq n$ is the set
\begin{equation*}
\Gab{n,k} \triangleq\Big\lbrace \; \big[\begin{matrix}
f(g_1) & \dots & f(g_n)
\end{matrix}\big] : f \in \Lsmallerk \;\Big\rbrace \subseteq \Fqm^n.
\end{equation*}
%where the fixed elements $g_1, \dots, g_n \in \Fqm$ are linearly independent over $\Fq$. 
\end{definition}
Note that the codewords can be seen as matrices in $\Fq^{m \times n}$.

\section{Decoding of Gabidulin Codes}\label{sec:decoding}

This section recalls how to decode errors and erasures with Gabidulin codes from \cite[Section~3.2.3]{wachter2013decoding}, shows which operations on linearized polynomials are required to be fast and which degrees the involved polynomials have.

Let $\c \in \Gab{n,k}$ be a codeword with corresponding information polynomial $f$, $\e \in \Fqm^n$ an error word and $\r = \c+\e$ the received word.
The decoding problem is to recover $\c$ from $\r$ if the rank of $\e$ is not too large.

If nothing about $\e$ is known, we say that \emph{only errors} occurred.
However, especially in applications like \emph{random linear network coding} \cite{silva2008rank}, $\e$ is partly known.
In particular, we can decompose $\e$ into
\begin{align*}
\e = \aE \BE + \aR \BR + \aC \BC,
\end{align*}
where the fragments correspond to 
\begin{itemize}
\item $\rk(\aE \BE) = \tE$ \emph{full errors}: $\aE \in \Fqm^\tE$, $\BE \in \Fq^{\tE \times n}$
\item $\rk(\aR \BR) = \tR$ \emph{row erasures}: $\aR \in \Fqm^\tR$, $\BR \in \Fq^{\tR \times n}$
\item $\rk(\aC \BC) = \tC$ \emph{column erasures}: $\aC \in \Fqm^\tC$, $\BC \in \Fq^{\tC \times n}$
\end{itemize}
and $\aR$ and $\BC$ are known at the receiver.
Note that if $\e$ and its fragments are interpreted as a matrices, $\ext(\aR)$ is a basis of the column space of $\ext(\aR \BR)$ and $\BC$ is a basis of the row space of $\ext(\aC \BC)$.
Using $\r$, $\aR$ and $\BC$, the receiver can compute the polynomials
\begin{align*}
\lR &= \MSP{\LH{a_1\Rexp,\dots,a_\tR\Rexp}}, \;
\qtr{r} = \IP{\{(g_i,r_i)\}_{i=1}^{n}}, \text{ and } \\
\gC &= \MSP{\LH{d_1\Cexp,\dots,d_\tC\Cexp}}, \text{ with } d_i\Cexp = \smallsum_{j=1}^{n} \BM_{i,j}\Cexp \beta^{[j-1]},
\end{align*}
the \emph{full $q$-reverse} $\gCbar \in \Lset^{<m}$ of $\gC$ with coefficients
\begin{align}
\gCbar_i = (\gC_{-i \mymod m})^{[i]}, \quad i=0,\dots,m-1. \label{eq:gCbar}
\end{align}
and the polynomials
\begin{align}
\gCtilde &= \gCbar \mul x^{[\tC]} \mod (x^{[m]}-x^{[0]}), \label{eq:gCtilde}\\
\qtr{y} &= \lR \mul \qtr{r} \mul \gCtilde \mod (x^{[m]}-x^{[0]}). \notag
\end{align}
Moreover, we define the unknown \emph{error locator polynomial}
\begin{align*}
\lE = \MSP{\LH{\lR(a_1\Eexp),\dots,\lR(a_\tE\Eexp)}}.
\end{align*}

With the help of these definitions, we can state the following key equation.
In the error and erasure case ($\tR>0$ or $\tC>0$), it only holds for $n=m$ and the $g_i$'s being a normal basis ($g_i = \beta^{[i-1]}$)\footnote{If $\tR=\tC=0$ (errors only), $\gC=\gCtilde=\lR=x^{[0]}$, and we obtain an ordinary key equation for Gabidulin codes (cf. \cite[Theorem~3.6]{wachter2013decoding}), which holds for arbitrary $g_i$'s and $n\leq m$, by replacing $(x^{[m]}-x^{[0]})$ by $\MSP{\LH{g_1,\dots,g_n}}$.}.
However, this does not appear to be a major disadvantage since e.g. we can use interleaving to obtain non-square matrices as codewords.
\begin{theorem}[{\cite[Theorem~3.8]{wachter2013decoding} and thereafter}]\label{thm:key_equation}
\begin{align*}
\lE \mul \qtr{y} \equiv \lE \mul \lR \mul f \mul \gCtilde \mod (x^{[m]}-x^{[0]})
\end{align*}
\end{theorem}

\subsection{Decoding Algorithm}

\begin{theorem}
If $2 \tE + \tR + \tC \leq d-1 = n-k$, Algorithm~\ref{alg:gaoalgo_ee} finds the correct information polynomial $f$.
\end{theorem}

\begin{proof}
Since $\qtr{y}$ is known and (cf. Table~\ref{tab:degrees})
\begin{align*}
\qdeg (\lE \lR f \gCtilde) < \lfloor \tfrac{n-k-\tR-\tC}{2} \rfloor + \tR + k + \tC = \lfloor \tfrac{n+k+\tR+\tC}{2} \rfloor,
\end{align*}
we can use the LEEA to obtain
\begin{align*}
[\rout,\uout,\vout] = \ALEEA{\qtr{y}, x^{[m]}-x^{[0]}, \lfloor\tfrac{n+k+\tR+\tC}{2}\rfloor }
\end{align*}
with $\rout = \uout \mul \qtr{y} + \vout \mul (x^{[m]}-x^{[0]})$ and $\qdeg \rout < \lfloor\tfrac{n+k+\tE+\tR}{2}\rfloor$.
It is shown in \cite{wachter2013decoding} that if $2 \tE + \tR + \tC \leq n-k$,
\begin{align*}
\uout = \lE \text{ and }
\rout = \lE \mul \lR \mul f \mul \gCtilde.
\end{align*}
Hence, we can obtain the evaluation polynomial $f$ by left-dividing $\rout$ by $\uout \mul \lR$ and then right-dividing it by $\gCtilde$.
\end{proof}

\subsection{Degrees of Involved Polynomials}

The degrees of the polynomials defined in this section are summarized in Table~\ref{tab:degrees}.
Since $\tE,\tR,\tC \leq n = m$, the following lemma is correct.
\begin{lemma}
All polynomials used in Alg.~\ref{alg:gaoalgo_ee} have $\qdeg \in \BigO{n}$.
\end{lemma}
This statement also implies \cite[Remark~8]{puchinger2015fast}, which holds for non-degenerate cases (i.e. $\tE,\tR,\tC \in \Theta(n)$, or in the errors-only case by using a different algorithm).

\begin{table}[h]
\caption{$q$-degrees of Polynomials used in \cref{alg:gaoalgo_ee}}
\label{tab:degrees}
\centering
\begin{tabular}{l|l|l}
\toprule
a 	& $\qdeg a$ 		& Reason \\
\midrule
$\qtr{r}$ 	& $< n$ 		& Interpolation at $n$ points. \\
$\gC$ 		& $= \tC$ 		& $\dim(\LH{d_1\Cexp,\dots,d_\tC\Cexp}) = \tC$. \\
%$\gCbar$ 	& $< m = n$		& By definition. \\
$\gCtilde$	& $\leq \tC$	& $\gCtilde_i = \gC_{(\tC-i) \mymod m} = 0 \; \forall \, i>\tC$. \\
$\qtr{y}$	& $< m$			& Reduced modulo $(x^{[m]}-x^{[0]})$. \\
$\lE$		& $\leq \tE$	& $\dim(\LH{\lR(a_1\Eexp),\dots,\lR(a_\tE\Eexp)}) \leq \tE$. \\
\bottomrule
\end{tabular}
\end{table}

\subsection{Required Operations on Linearized Polynomials}

It was shown on \cite{wachter2013decoding} that the LEEA with polynomials in $\Lsetmaxs$ requires $\log(s)$ many divisions.
Using this, the operations on linearized polynomials used in \cref{alg:gaoalgo_ee} are outlined in \cref{tab:operations}, together with a notation for the respective complexity.
\begin{table}[h]
\caption{Operations used in \cref{alg:gaoalgo_ee}}
\label{tab:operations}
\centering
\begin{tabular}{l|c}
\toprule
Operation ($a,b \in \Lsetmaxs$, $U \subseteq \Fqm$: $|U| \leq s$) 	& Complexity Notation \\
\midrule
Multiplication $a \mul b$					& $\OMul{s}$ \\
Right (or left) division of $a$ by $b$		& $\ODiv{s}$ \\
Calculation of $\MSP{\LH{U}}$					& $\OMSP{s}$ \\
MPE of $a$ at elements of $U$				& $\OMPE{s}$ \\
Interpolation at $\leq s$ point tuples		& $\OIP{s}$ \\
\bottomrule
\end{tabular}
\end{table}

Hence, the decoding complexity is directly determined by these operations.
The next section shows that they can all be accomplished in sub-quadratic time in $s$.

%%%%%%%%%%%%%%%%%%%%%%%%%%%%%%%%%%%%%%%%%%%%%%%%%%%%%%%%%%%%%%%%%%%%%%%%%%%%%%
% Known Algorithms
%%%%%%%%%%%%%%%%%%%%%%%%%%%%%%%%%%%%%%%%%%%%%%%%%%%%%%%%%%%%%%%%%%%%%%%%%%%%%%
\section{Fast Algorithms}\label{sec:fast_algos}

In this section, we present fast multiplication and division algorithms in $\Fqm[x;\autom]$ and methods for MPE, calculation of MSPs and interpolation in $\Lset$ with subquadratic complexity.
Complexities are counted in operations in $\Fqm$.
All algorithms and proofs are presented in full detail in the extended version of this paper \cite{puchinger2015fast}.
Here, we give brief summaries in order to outline proof ideas.

\subsection{Fast Multiplication}\label{subsec:fast_mul}

We generalize the fast multiplication algorithm for linearized polynomials from \cite[Theorem~3.1]{wachter2013decoding} to skew polynomials.
This generalization is needed for the division algorithm in Section~\ref{subsec:fast_div}.
We consider polynomials $a,b \in \Asetmaxs$ and define $\sstar := \lceil \sqrt{s+1} \rceil$.

\begin{theorem}\label{thm:multiplication}
If $\automi{i}(\alpha)$ can be computed in $\BigO{1}$ over $\Fqm$, the multiplication of $a,b \in \Asetmaxs$ using Algorithm~\ref{alg:Mul} costs
\begin{align*}
\OMulSkew{s} \in \BigO{s^{\frac{\omega+1}{2}}}.
\end{align*}
\end{theorem}

\begin{proof}
See \cite{puchinger2015fast}.
The proof uses a fragmentation of $a$ into
\begin{align*}
a^{(i)} = \smallsum_{j=0}^{\sstar-1} a_{i \sstar + j} x^{i \sstar + j}
\end{align*}
and $c^{(i)} := a^{(i)} \amul b$. Then the polynomial multiplication is reduced to matrix multiplication involving the following matrices (cf. \cref{alg:Mul}).
\begin{align}
C &= \begin{bmatrix}
C_{ij}
\end{bmatrix}_{i=0,\dots,\sstar-1}^{j=0,\dots,s+\sstar-1},
\;  C_{ij} = \automi{-i \sstar}(c_j^{(i)}),  \notag \\
A &= \begin{bmatrix}
A_{ij}
\end{bmatrix}_{i=0,\dots,\sstar-1}^{j=0,\dots,\sstar-1},
\quad \; A_{ij} = \automi{-i \sstar}(a_{i \sstar+j}), \label{eq:mul_matrices}\\
B &= \begin{bmatrix}
B_{ij}
\end{bmatrix}_{i=0,\dots,\sstar-1}^{j=0,\dots,s+\sstar-1},
\; B_{ij} = \begin{cases}
\automi{j}(b_{i-j}), &0 \leq i-j \leq s, \\
0, &\text{else}. \hfill \qedhere
\end{cases} \notag 
\end{align}
\end{proof}

\vspace{-0.3cm}

\printalgoIEEE{
%\begin{algorithm}
\DontPrintSemicolon
\KwIn{$a,b \in \Asetmaxs$}
\KwOut{$c = a \amul b$}
Set up matrices $A$ and $B$ as in~\eqref{eq:mul_matrices} \hfill \tcp{$s^{\frac{3}{2}} \cdot\BigO{1}$} \label{line:Mul_AB}
$C \gets A \cdot B$ \hfill \tcp{$\sstar \cdot \BigO{{\sstar}^\omega}$}  \label{line:Mul_MatrixMul}
Extract the $c^{(i)}$'s from $C$ as in~\eqref{eq:mul_matrices} \hfill \tcp{$s^{\frac{3}{2}} \cdot\BigO{1}$}  \label{line:Mul_C}
\Return{$c \gets \sum_{i=0}^{\sstar-1} c^{(i)}$} \hfill \tcp{$\BigO{s^{\frac{3}{2}}}$}  \label{line:Mul_Add}
\caption{Multiplication}
\label{alg:Mul}
%\end{algorithm}
}

If $\autom \in \mathrm{Gal}(\Fqm/\Fq)$, $\autom$ is of the form $\cdot^{[j]}$ for some $j$.
If elements of $\Fqm$ are represented in a normal basis over $\Fq$, then $\autom^i(\alpha) = \alpha^{[i+j]}$ can be computed in $\BigO{1}$ by a cyclic shift of the coefficient vector (cf. \cite{wachter2013decoding}).
Thus, Theorem \ref{thm:multiplication} holds for all $\Fqm[x;\autom]$ and with $\omega \approx 2.376$ it follows that
\begin{align*}
\OMul{s} \in \BigO{s^{1.69}}.
\end{align*}

\subsection{Fast Division}\label{subsec:fast_div}

It was shown in \cite[Section 2.1.2]{caruso2012some} that division in a skew polynomial ring $\Fqm[x;\autom]$ can be reduced to multiplication in $\Fqm[x;\autom^{-1}]$.
Together with \cref{alg:Mul}, we obtain a fast division algorithm for $\Lset \cong \Fqm[x;\cdot^q]$.
Since the multiplication algorithm of Section~\ref{subsec:fast_mul} was so far only known for linearized polynomials, it was not obvious how to combine these results.
We only consider right division in this chapter and the left division works analogously.
To describe the algorithm, we need the following bijective mapping and corresponding lemmas:
\begin{align*}
\tau_s : \Asetmaxs &\to \Asetinvmaxs \\
a = \smallsum_{i=0}^{s} a_i x^i &\mapsto \tau_s(a) = \smallsum_{i=0}^{s} a_{s-i} x^{i}.
\end{align*}

\begin{lemma}[\cite{caruso2012some}]\label{lem:tau_div}
Let $\quo,\remainder \in \Aset$ quotient and remainder of the right division of $a \in \Aset$ by $b \in \Aset$ with $s = \deg a \geq \deg b = \ell$. Then, with $b^{(s-\ell)} := \sum_{i=0}^{\ell} \automi{s-\ell}(b_i) x^i$,
\begin{align*}
\tau_s(a) \equiv \tau_{s-\ell}(\quo) \amul \tau_\ell(b^{(s-\ell)}) \mod x^{s-\ell+1}.
\end{align*}
\end{lemma}

\begin{lemma}[\cite{caruso2012some}]\label{lem:hensel_inverse}
The right inverse of $\tau_\ell(b^{(s-\ell)})$ modulo $x^{s-\ell+1}$ exists and can be calculated by Algorithm \ref{alg:hensel_inverse} in $\BigO{\OMulSkew{s} \log s}$ time.
\end{lemma}

\printalgoIEEE{
\DontPrintSemicolon
\KwIn{$c \in \Fqm[x;\autom^{-1}]$ with $c_0 \neq 0$, $k \in \N$.}
\KwOut{$d \in \Fqm[x;\autom^{-1}]$ s.t. $c \amul d \equiv 1 \mod x^k$}
$h_0 \gets 1/c_0$ \hfill \tcp{$\BigO{1}$} \label{line:inv_init}
\For{$i=1,\dots,\lceil \log_2(k) \rceil$}{
$h_{i} \gets 2 h_{i-1}-h_{i-1} \amul c \amul h_{i-1} \mod x^{2^i}$  \hfill \tcp{$\OMul{s^{2^i}}$} \label{line:inv_loop}
}
\Return{$h_{\lceil \log_2(k) \rceil}$} \label{line:inv_return}
\caption{$\ARI{c,k}$ {\cite[Algorithm~1]{caruso2012some}}}
\label{alg:hensel_inverse}
}

The following theorem shows the reduction of the skew polynomial division to skew polynomial multiplication.
\begin{theorem}\label{thm:division-compl}
$\ODivSkew{s} \in \BigO{\OMulSkew{s} \log s}$ using Alg.~\ref{alg:division}.
\end{theorem}

\begin{proof}
Lemma~\ref{lem:tau_div} implies the correctness.
Line~\ref{line:div_hensel} is the complexity bottleneck and can be accomplished in $\BigO{\OMulSkew{s} \log s}$ according to Lemma~\ref{lem:hensel_inverse}.
\end{proof}

\printalgoIEEE{
%\begin{algorithm}
\DontPrintSemicolon
\KwIn{$a,b \in \Aset$, $s = \deg a \geq \deg b = \ell$}
\KwOut{$\quo, \remainder \in \Aset$ s.t. $a = \quo \amul b + \remainder$ and $\deg \remainder < \ell$.}
$c \gets \tau_\ell(b^{(s-\ell)})$; $\,$ $\tilde{a} \gets \tau_s(a)$ \hfill \tcp{$\BigO{s}$} \label{line:div_c}
$c^{-1} \gets \ARI{c,s-\ell+1}$ \hfill \tcp{$\!\!\BigO{\OMul{s}\log s}$} \label{line:div_hensel}
$\quo \gets \tau_{s-\ell}^{-1}\big(\tilde{a} \amul c^{-1} \mod x^{\ell-1}\big)$ \hfill \tcp{$\OMul{s}$} \label{line:div_q}
$\remainder \gets a-\quo \amul b$ \hfill \tcp{$\OMul{s}$} \label{line:div_r}
\Return{$[\quo,\remainder]$} \label{line:div_return}
\caption{$\ADIV{a,b}$ {\cite[Algorithm~1]{caruso2012some}}}
\label{alg:division}
%\end{algorithm}
}

% Table on last page (bug: needs to be defined on second last page!)
\begin{table*}[!b]
    \caption{New complexity bounds over $\Fqm$ for operations with linearized polynomials}
    \label{tab:overview_operations_linearized_new}
    \centering
	\begin{tabular}{p{0.09\textwidth} | p{0.21\textwidth} p{0.16\textwidth} p{0.08\textwidth} | p{0.08\textwidth} p{0.20\textwidth}}
		\toprule	
		Operation	& New Complexity (exact) & $\omega \approx 2.376$ & Source	& Before & Source\\
		\midrule
		$\OMul{s}$
		&
		$\BigO{s^{\frac{\omega+1}{2}}}$
		&
		$\BigO{s^{1.69}}$
		&
		\cite{wachter2013decoding}
		&
		$\BigO{s^{1.69}}$
		&
		\cite{wachter2013decoding}
		\\[1ex]
		$\ODiv{s}$
		&
		$\BigO{\OMulSkew{s} \log s}$
		&
		$\BigO{s^{1.69}\log s}$
		&
		\cref{thm:division-compl}
		&
		$\BigO{s^2 \log(s)}$
		&
		\cite{caruso2012some}
		\\[1ex]
		$\OMSP{s}$
		&
		$\BigO{s^{\max\{\log_2(3), \frac{\omega+1}{2}\}} \log^2(s)}$
		&
		$\BigO{s^{1.69} \log^2(s)}$
		&
		\cref{thm:mspmpe-compl}
		&
		$\BigO{s^2}$
		&
		\cite{silva2008rank}
		\\[1ex]
		$\OMPE{s}$
		&
		$\BigO{s^{\max\{\log_2(3), \frac{\omega+1}{2}\}} \log^2(s)}$
		&
		$\BigO{s^{1.69} \log^2(s)}$
		&
		\cref{thm:mspmpe-compl}
		&
		$\BigO{s^2}$
		&
		``naive'' ($s$ ordinary evaluations)
		\\[1ex]
		$\OIP{s}$
		&
		$\BigO{\OMSP{s}}$
		&
		$\BigO{s^{1.69} \log^2(s)}$
		&
		Theorem~\ref{thm:interpol-comp}
		&
		$\BigO{s^3}$
		&
		``naive'' (Lagrange bases \cite{silva2007rank})
		\\
		\bottomrule
    \end{tabular}
\end{table*}

\subsection{Fast Computation of MSP and MPE}\label{subsec:fast_msp}

The fast algorithm for MPE requires a call of the fast algorithm for calculating the MSP and vice versa and therefore, their complexities have to be analyzed jointly.
The following two lemmas show important relations between the MPE and the MSP.
%Lemma~\ref{lem:MSP_recursion} was already proven in~\cite{li2014transform}, but for the sake of completeness, we give a slightly more detailed proof.

\begin{lemma}[\cite{li2014transform}]\label{lem:MSP_recursion}
Let $U = \{u_1,\dots,u_s\}$ be a basis of a subspace $\U \subseteq \Fqm$, $A,B \subseteq \Fqm$ s.t. $U = A \cup B$. Then,
\begin{align}
\MSP{\U} &= \MSP{\LH{U}} = \MSP{\LH{\MSP{\LH{A}}(B)}} \mul \MSP{\LH{A}} \text{ and} \notag \\
M_{\LH{u_i}} &= \begin{cases}
x^{[0]}, & \text{if } u_i=0, \\
x^{[1]} - u_i^{q-1} x^{[0]}, & \text{else.}
\end{cases} \label{eq:MSP_recursion_base}
\end{align}
\end{lemma}

\begin{lemma}\label{lem:MPE_recursion}
Let $a \in \Lset$ and let $U, A,B \subseteq \Fqm$ where $A,B\subseteq \Fqm$ are disjoint and $U = A \cup B$. 
Let $\remainder_A,\remainder_B$ be the remainders of the right divisions of $a$ by $\MSP{\LH{A}}$ and $\MSP{\LH{B}}$ respectively.
Then, the MSP of $a$ at the set $U$ is
\begin{align*}
a(U) = \remainder_A(A) \cup \remainder_B(B).
\end{align*}
If $U = \{u\}$ and $\qdeg a \leq 1$, $a(U) = \{a(u) = a_1 u^{[1]} + a_0 u^{[0]} \}$.
\end{lemma}

This implies the main statement of this subsection.
\begin{theorem}\label{thm:mspmpe-compl}
MSP and MPE can be calculated with Algorithm~\ref{alg:MSP} and \ref{alg:MPE} in complexity
$\OMSP{s}$ and $\OMPE{s} \in \BigO{s^{\max\{\log_2(3), \frac{\omega+1}{2}\}} \log^2(s)} \subseteq \BigO{s^{1.69} \log^2(s)}$.
\end{theorem}

\begin{proof}
See \cite{puchinger2015fast}.
Correctness follows from Lemma~\ref{lem:MSP_recursion} and \ref{lem:MPE_recursion}.
Complexity-wise, we can prove the system of recursion
\begin{align*}
\begin{bmatrix}
\OMSP{s} \\
\OMPE{s}
\end{bmatrix}
&=
\begin{bmatrix}
2 & 1 \\
1 & 2
\end{bmatrix}
\cdot
\begin{bmatrix}
\OMSP{\tfrac{s}{2}} \\
\OMPE{\tfrac{s}{2}}
\end{bmatrix}
+
\begin{bmatrix}
\OMul{s} \\
2 \cdot \ODiv{s}
\end{bmatrix}.
\end{align*}
Thus, the complexities $\OMSP{s}$ and $\OMPE{s}$ depend on the maximum eigenvalue $\lambda = 3$ of the system's matrix and the complexities $\OMul{s}$ and $\ODiv{s}$, proving the claim.
\end{proof}

\printalgoIEEE{
%\begin{algorithm}
\DontPrintSemicolon
\KwIn{Basis $U = \{u_1, \dots, u_s\}$ of a subspace $\U \subseteq \Fqm$.}
\KwOut{MSP $\MSP{\langle U \rangle}$.}
\lIf{$s=1$}{\Return{$M_{\langle u_1\rangle}(x)$ according to \eqref{eq:MSP_recursion_base}}} \label{line:MSP_base_case}
\Else{
$A \gets \{u_1,\dots,u_{\gauss{\frac{s}{2}}}\}$, $B \gets \{u_{\gauss{\frac{s}{2}}+1},\dots,u_s\}$ \hfill \tcp{$\BigO{1}$} \label{line:MSP_partition}
$\MSP{\langle A \rangle} \gets \AMSP{A}$ \hfill \tcp{$\OMSP{\tfrac{s}{2}}$}  \label{line:MSP_MSPA}
$\MSP{\langle A \rangle}(B) \gets \AMPE{\MSP{\langle A \rangle}}{B}$  \hfill \tcp{$\OMPE{\tfrac{s}{2}}$}  \label{line:MSP_MPEB}
$\MSP{\MSP{\langle A \rangle}(B) \rangle} \gets \AMSP{\MSP{\langle A \rangle}(B)}$ \hfill \tcp{$\OMSP{\tfrac{s}{2}}$}  \label{line:MSP_MSPB}
\Return{$\MSP{\MSP{\langle A \rangle}(B) \rangle} \mul \MSP{\langle A \rangle}$} \hfill \tcp{$\OMul{s}$}  \label{line:MSP_mul}
}
\caption{$\AMSP{U}$}
\label{alg:MSP}
%\end{algorithm}
}

\printalgoIEEE{
%\begin{algorithm}
\DontPrintSemicolon
\KwIn{$a\in \Lsetmaxs$, $\{u_1, \dots,u_s\} \in \Fqm^s$}
\KwOut{Evaluation of $a$ at all points $u_i$}
\lIf{$s=1$}{\Return{$\{a_1 u_1^{[1]} + a_0 u_1^{[0]}\}$}} \label{line:MPE_base_case}
\Else{
$A\! \gets\!\! \{u_1,\dots,u_{\gauss{\frac{s}{2}}}\}$, $B \!\gets\!\! \{u_{\gauss{\frac{s}{2}}+1},\dots,u_s\}$ \hfill \tcp{$\!\!\!\BigO{1}$}\label{line:MPE_partition}
$\MSP{\langle A \rangle} \gets \AMSP{A}$ \hfill \tcp{$\OMSP{\tfrac{s}{2}}$} \label{line:MPE_MSPA}
$\MSP{\langle B \rangle} \gets \AMSP{B}$ \hfill \tcp{$\OMSP{\tfrac{s}{2}}$} \label{line:MPE_MSPB}
$[\quo_A, \remainder_A] \gets \ARDIV{a, \MSP{\langle A \rangle}}$ \hfill \tcp{$\ODiv{s}$} \label{line:MPE_div_A}
$[\quo_B, \remainder_B] \gets \ARDIV{a, \MSP{\langle B \rangle}}$ \hfill \tcp{$\ODiv{s}$} \label{line:MPE_div_B}
\Return{$\AMPE{\remainder_A}{A} \cup \AMPE{\remainder_B}{B}$} \hfill \tcp{$2 \cdot \OMPE{\tfrac{s}{2}}$} \label{line:MPE_MPE}
}
\caption{$\AMPE{a}{\{u_1, \dots,u_s\}}$}
\label{alg:MPE}
%\end{algorithm}
}

\subsection{Fast Interpolation}\label{subsec:fast_interp}

This subsection shows that linearized interpolation can be reduced to calculating MSPs and MPEs and therefore, our fast algorithms from the previous subsection can be applied.

\begin{lemma}\label{lem:interpolation_recursion}
For the interpolation polynomial, it holds that
\begin{align*}
\IP{\{(x_i,y_i)\}_{i=1}^{s}} &= \IP{\{(\tilde{x}_i,y_i)\}_{i=1}^{\gauss{\frac{s}{2}}}} \mul \MSP{\LH{x_{\gauss{\frac{s}{2}}+1}, \dots, x_s}} \\
&+ \IP{\{(\tilde{x}_i,y_i)\}_{i=\gauss{\frac{s}{2}}+1}^{s}} \mul \MSP{\LH{x_1,\dots,x_{\gauss{\frac{s}{2}}}}} \\
\text{with } \tilde{x}_i &:= \begin{cases}
\MSP{\LH{x_{\gauss{\frac{s}{2}}+1}, \dots, x_s}}(x_i), &\text{if } i=1,\dots,\gauss{\frac{s}{2}} \\
\MSP{\LH{x_1,\dots,x_{\gauss{\frac{s}{2}}}}}(x_i), &\text{otherwise}
\end{cases}
\end{align*}
and $\IP{\{(x_i,y_i)\}_{i=1}^{1}} = \frac{y_1}{x_1} x^{[0]}$ (base case $s=1$).
\end{lemma}

\begin{proof}
See \cite{puchinger2015fast}. The idea is to evaluate $\IP{\{(x_i,y_i)\}_{i=1}^{s}}$ at all positions $x_i$ and show that the definition holds.
\end{proof}

\begin{theorem}\label{thm:interpol-comp}
$\OIP{s} \in \BigO{\OMSP{s}}$ using \cref{alg:IP}.
\end{theorem}

\begin{proof}
Correctness follows from \cref{lem:interpolation_recursion}.
The complexity is $\OIP{s} = 2 \cdot \OIP{\tfrac{s}{2}} + \BigO{\OMSP{s}}$, which is resolved using the master theorem, implying the claim.
\end{proof}

\printalgoIEEE{
%\begin{algorithm}
\DontPrintSemicolon
\KwIn{$(x_1,y_1),\dots,(x_s,y_s) \in \Fqm^2$, $x_i \neq 0$ distinct}
\KwOut{Interpolation polynomial $\IP{\{(x_i,y_i)\}_{i=1}^{s}}$}
\lIf{$s=1$}{\Return{$\{\frac{y_1}{x_1} x^{[0]}\}$}} %\hfill \tcp{$\BigO{1}$}} \label{line:IP_base_case}
\Else{
$A \gets \{x_1,\dots,x_{\gauss{\frac{s}{2}}}\}$, $B \gets \{x_{\gauss{\frac{s}{2}}+1},\dots,x_s\}$ \hfill \tcp{$\BigO{1}$}\label{line:IP_partition}
%$\MSP{\langle A \rangle} \gets \AMSP{A}$, $\MSP{\langle B \rangle} \gets \AMSP{B}$ \hfill \tcp{$2 \cdot \OMSP{\tfrac{s}{2}}$} \label{line:IP_MSP}
$\MSP{\langle A \rangle} \gets \AMSP{A}$ \hfill \tcp{$\OMSP{\tfrac{s}{2}}$} \label{line:IP_MSPA}
$\MSP{\langle B \rangle} \gets \AMSP{B}$ \hfill \tcp{$\OMSP{\tfrac{s}{2}}$} \label{line:IP_MSPB}
$\{\tilde{x}_1,\dots,\tilde{x}_{\gauss{\frac{s}{2}}}\} \gets \AMPE{\MSP{\langle B \rangle}}{A}$ \hfill \tcp{$\OMPE{\tfrac{s}{2}}$} \label{line:IP_x_tilde_1}
$\{\tilde{x}_{\gauss{\frac{s}{2}}+1},\dots,\tilde{x}_s\} \gets \AMPE{\MSP{\langle A \rangle}}{B}$ \hfill \tcp{$\OMPE{\tfrac{s}{2}}$} \label{line:IP_x_tilde_2}
$\IP{1} \gets \AIP{\{(\tilde{x}_i,y_i)\}_{i=1}^{\gauss{\frac{s}{2}}}}$ \hfill \tcp{$\OIP{\tfrac{s}{2}}$} \label{line:IP_IP1}
$\IP{2} \gets \AIP{\{(\tilde{x}_i,y_i)\}_{i=\gauss{\frac{s}{2}}+1}^{s}}$ \hfill \tcp{$\OIP{\tfrac{s}{2}}$} \label{line:IP_IP2}
\Return{$\IP{1} \mul \MSP{\langle B \rangle} + \IP{2} \mul \MSP{\langle A \rangle}$} \hfill \tcp{$2 \cdot \OMul{\tfrac{s}{2}}$} \label{line:IP_mul}
}
\caption{$\AIP{\{(x_i,y_i)\}_{i=1}^{s}}$}
\label{alg:IP}
%\end{algorithm}
}

\subsection{Comparsion to Other Fast Algorithms}\label{subsec:comparison}

In~\cite{SilvaKschischang-FastEncodingDecodingGabidulin-2009} and \cite{WachterAfanSido-FastDecGabidulin_DCC}, several operations with linearized polynomials $\Lset$ with degree $\leq m$ were reduced to complexity $\BigO{m^3}$ in operations in $\Fq$.
It is shown in \cite{couveignes2009elliptic} that for any field extension $\Fqm/\Fq$, there is a representation of $\Fqm$ elements over $\Fq$ such that the operations addition, multiplication and Frobenius powering $\cdot^{q}$ with $\Fqm$ elements cost
\begin{align*}
\BigO{m \log^3(m) \log(\log(m))^3}
\end{align*}
operations in $\Fq$.
Hence, our algorithms have complexity
\begin{align*}
\BigO{m^{2.69} \log^5(m) \log(\log(m))^3}
\end{align*}
over $\Fq$ and improve the results of \cite{SilvaKschischang-FastEncodingDecodingGabidulin-2009} and \cite{WachterAfanSido-FastDecGabidulin_DCC}.

\section{Main Statement}

By combining our analysis of the error and erasure decoding algorithm for Gabidulin codes in Section~\ref{sec:decoding} with the fast operations presented in Section~\ref{sec:fast_algos}, which are summarized in Table~\ref{tab:overview_operations_linearized_new}, we obtain the following main statement of the paper.
\begin{theorem}
Error and erasure decoding with a Gabidulin code $\Gab{n,k}$ has complexity
\begin{align*}
\BigO{n^{1.69} \log^2(n)} \text{ in } \Fqm.
\end{align*}
\end{theorem}

Note that encoding $\Lset^{<k} \to \Fqm^n, f \mapsto (f(g_1),\dots,f(g_n))$ of Gabidulin codes is a multi-point evaluation and can also be accomplished in $\BigO{n^{1.69} \log^2(n)}$ time.

For future work, it is interesting to include our new algorithms in the study from~\cite{BohaczukSilva-EvaluationErasureDecodingGabidulin} on fast erasure decoding of Gabidulin codes and generalize the results to skew polynomials over arbitrary fields.

\section*{Acknowledgement}
The authors would like to thank Johan~S.~R.~Nielsen for the valuable discussions and Luca~De~Feo for pointing us at \cite{couveignes2009elliptic}.

\bibliographystyle{IEEEtran}
\bibliography{main}

% Generated by IEEEtran.bst, version: 1.13 (2008/09/30)
\begin{thebibliography}{10}
\providecommand{\url}[1]{#1}
\csname url@samestyle\endcsname
\providecommand{\newblock}{\relax}
\providecommand{\bibinfo}[2]{#2}
\providecommand{\BIBentrySTDinterwordspacing}{\spaceskip=0pt\relax}
\providecommand{\BIBentryALTinterwordstretchfactor}{4}
\providecommand{\BIBentryALTinterwordspacing}{\spaceskip=\fontdimen2\font plus
\BIBentryALTinterwordstretchfactor\fontdimen3\font minus
  \fontdimen4\font\relax}
\providecommand{\BIBforeignlanguage}[2]{{%
\expandafter\ifx\csname l@#1\endcsname\relax
\typeout{** WARNING: IEEEtran.bst: No hyphenation pattern has been}%
\typeout{** loaded for the language `#1'. Using the pattern for}%
\typeout{** the default language instead.}%
\else
\language=\csname l@#1\endcsname
\fi
#2}}
\providecommand{\BIBdecl}{\relax}
\BIBdecl

\bibitem{silva_rank_metric_approach}
D.~Silva, F.~R. Kschischang, and R.~K{\"o}tter, ``{A Rank-Metric Approach to
  Error Control in Random Network Coding},'' \emph{IEEE Trans. Inform. Theory},
  vol.~54, no.~9, pp. 3951--3967, 2008.

\bibitem{Loidreau2010Designing}
P.~Loidreau, ``{Designing a Rank Metric Based McEliece Cryptosystem},'' in
  \emph{Post-Quantum Cryptography}, 2010, pp. 142--152.

\bibitem{SilbersteinRawatVish-ErrorResilDistributedStorage_2012}
N.~Silberstein, A.~S. Rawat, and S.~Vishwanath, ``{Error Resilience in
  Distributed Storage via Rank-Metric Codes},'' in \emph{Allerton Conf.
  Communication, Control, Computing (Allerton)}, Oct. 2012, pp. 1150--1157.

\bibitem{wachter2013decoding}
A.~{Wachter-Zeh}, ``{Decoding of Block and Convolutional Codes in Rank
  Metric},'' Ph.D. dissertation, Ulm University and University of Rennes, 2013.

\bibitem{Gadouleau_complexity}
M.~Gadouleau and Z.~Yan, ``{Complexity of Decoding Gabidulin Codes},'' in
  \emph{42nd Annual Conf. Inform. Sciences and Systems (CISS)}, Mar. 2008, pp.
  1081--1085.

\bibitem{caruso2012some}
X.~Caruso and J.~{Le Borgne}, ``{Some algorithms for skew polynomials over
  finite fields},'' \emph{arXiv preprint arXiv:1212.3582}, 2012.

\bibitem{li2014transform}
W.~Li, V.~Sidorenko, and D.~Silva, ``{On Transform-Domain Error and Erasure
  Correction by Gabidulin Codes},'' \emph{Designs, Codes and Cryptography},
  vol.~73, no.~2, pp. 571--586, 2014.

\bibitem{SilvaKschischang-FastEncodingDecodingGabidulin-2009}
D.~Silva and F.~R. Kschischang, ``{Fast Encoding and Decoding of Gabidulin
  Codes},'' in \emph{IEEE Int. Symp. Inf. Theory (ISIT)}, Jun. 2009, pp.
  2858--2862.

\bibitem{WachterAfanSido-FastDecGabidulin_DCC}
A.~Wachter-Zeh, V.~Afanassiev, and V.~Sidorenko, ``{Fast Decoding of
  {G}abidulin Codes},'' \emph{Designs, Codes and Cryptography}, vol.~66, no.~1,
  pp. 57--73, 2013.

\bibitem{puchinger2015fast}
S.~Puchinger and A.~Wachter-Zeh, ``Fast operations on linearized polynomials
  and their applications in coding theory,'' \emph{Submitted to J. of Symb.
  Comp.}, Dec. 2015, arXiv preprint \url{http://arxiv.org/abs/1512.06520}.

\bibitem{ore1933special}
O.~Ore, ``{On a Special Class of Polynomials},'' \emph{Transactions of the
  American Mathematical Society}, vol.~35, no.~3, pp. 559--584, 1933.

\bibitem{ore1933theory}
------, ``{Theory of Non-Commutative Polynomials},'' \emph{Annals of
  Mathematics}, vol.~34, no.~3, pp. 480--508, Jul. 1933.

\bibitem{Delsarte_1978}
P.~Delsarte, ``{Bilinear Forms over a Finite Field with Applications to Coding
  Theory},'' \emph{J. Combin. Theory Ser. A}, vol.~25, no.~3, pp. 226--241,
  1978.

\bibitem{Gabidulin_TheoryOfCodes_1985}
E.~M. Gabidulin, ``{Theory of Codes with Maximum Rank Distance},'' \emph{Probl.
  Inf. Transm.}, vol.~21, no.~1, pp. 3--16, 1985.

\bibitem{Roth_RankCodes_1991}
R.~M. Roth, ``{Maximum-Rank Array Codes and their Application to Crisscross
  Error Correction},'' \emph{IEEE Trans. Inform. Theory}, vol.~37, no.~2, pp.
  328--336, 1991.

\bibitem{silva2008rank}
D.~Silva, F.~R. Kschischang, and R.~Koetter, ``{A Rank-Metric Approach to
  CError Control in Random Network Coding},'' \emph{IEEE Transactions on
  Information Theory}, vol.~54, no.~9, pp. 3951--3967, 2008.

\bibitem{silva2007rank}
D.~Silva and F.~R. Kschischang, ``{Rank-Metric Codes for Priority Encoding
  Transmission with Network Coding},'' in \emph{Canadian Workshop on
  Information Theory}.\hskip 1em plus 0.5em minus 0.4em\relax IEEE, 2007, pp.
  81--84.

\bibitem{couveignes2009elliptic}
J.-M. Couveignes and R.~Lercier, ``{Elliptic Periods for Finite Fields},''
  \emph{{Finite Fields and Their Applications}}, vol.~15, no.~1, pp. 1--22,
  2009.

\bibitem{BohaczukSilva-EvaluationErasureDecodingGabidulin}
R.~Bohaczuk~Venturelli and D.~Silva, ``{An evaluation of erasure decoding
  algorithms for Gabidulin codes},'' in \emph{Int. Telecommunications Symposium
  (ITS)}, Aug. 2014, pp. 1--5.

\end{thebibliography}

\end{document}